\documentclass[reqno]{amsproc}
\usepackage{tikz}
\usepackage{color}
\usepackage{amsmath}
\usepackage{amsfonts}
\usepackage{geometry}
\usepackage{amssymb}
\usepackage{bbm}
\usepackage{mathrsfs}
\usepackage{pbox}
\usepackage[utf8x]{inputenc}
\usepackage[loose]{subfigure}

\usepackage[all]{xy}

\definecolor{myurlcolor}{rgb}{0,0,0.7}
\definecolor{myrefcolor}{rgb}{0.8,0,0}
\usepackage{hyperref}
\hypersetup{colorlinks,
linkcolor=myrefcolor,
citecolor=myurlcolor,
urlcolor=myurlcolor}

\newcommand{\C}{\mathbb{C}}
\newcommand{\N}{\mathbb{N}}

\newcommand{\R}{\mathbb{R}}
\newcommand{\ra}{\rightarrow}

\newcommand{\cone}{\mathrm{cone}}
\newcommand{\omin}{\otimes_{\min}}
\newcommand{\omax}{\otimes_{\max}}
\newcommand{\lin}{\mathrm{lin}}
\newcommand{\tr}{\mathrm{tr}}
\newcommand{\id}{\mathrm{id}}
\newcommand{\ex}{\mathrm{ex}}

\theoremstyle{plain}
\newtheorem{thm}{Theorem}
\newtheorem{lem}[thm]{Lemma}
\newtheorem{prop}[thm]{Proposition}
\newtheorem{cor}[thm]{Corollary}
\newtheorem{qstn}[thm]{Question}

\newtheorem{defn}[thm]{Definition}

\theoremstyle{definition}
\newtheorem{expl}[thm]{Example}

\theoremstyle{remark}

\newcommand{\beq}{\begin{equation}}
\newcommand{\eeq}{\end{equation}}

\begin{document}



\title[Duality in Bell scenarios]{Polyhedral duality in Bell scenarios with two binary observables}

\author{Tobias Fritz}
\address{ICFO--Institut de Ciencies Fotoniques\\ 
Mediterranean Technology Park\\ 
08860 Castelldefels (Barcelona)\\ 
Spain}
\email{tobias.fritz@icfo.es}

\keywords{polyhedral duality; Bell inequalities; no-signaling boxes}

\subjclass[2010]{}

\thanks{\textit{Acknowledgements.} The author would like to thank Chirag Dhara for initiating this by pointing to~\cite{Sliwa}; Jean-Daniel Bancal for creating Table~\ref{correspond}; Antonio Ac\'in, Carlos Palazuelos, Ignacio Villanueva and Gonzalo de la Torre for helpful discussion; and the EU STREP QCS for financial support.}

\begin{abstract}
For the Bell scenario with two parties and two binary observables per party, it is known that the no-signaling polytope is the polyhedral dual (polar) of the Bell polytope. Computational evidence suggests that this duality also holds for three parties. Using ideas of Werner, Wolf, \.Zukowski and Brukner, we prove this for any number of parties by describing a simple linear bijection mapping (facet) Bell inequalities to (extremal) no-signaling boxes and vice versa. Furthermore, a symmetry-based technique for extending Bell inequalities (resp.~no-signaling boxes) with two binary observables from $n$ parties to $n+1$ parties is described; the Mermin--Klyshko family of Bell inequalities arises in this way, as well as $11$ of the $46$ classes of facet Bell inequalities for $3$ parties. Finally, we ask whether the set of quantum correlations is self-dual with respect to our transformation. We find this not to be the case in general, although it holds for $2$ parties on the level of correlations. This self-duality implies Tsirelson's bound for the CHSH inequality.
\end{abstract}

\maketitle

\section{Introduction}
\label{introduction}

\subsection*{Nonlocal correlations}

Since the groundbreaking work of Bell~\cite{Bell}, it is known that the correlations measurable in most quantum systems cannot be explained by classical probabilistic theories with only local interactions.

In more detail, a \emph{Bell scenario} is a triple of natural numbers $(n,k,l)$ where $n$ is the number of \emph{parties}, $k$ is the number of \emph{observables} per party and $l$ is the number of \emph{outcomes} of each observable of each party. We think of each party as a physicist experimenting in his own laboratory, and of the different laboratories as so far apart that any direct interaction between them can be excluded. Now each party receives a physical system---typically an atom or a photon---from a common source, and applies to it a certain measurement chosen at their own free will. The observables of party $j\in\{1,\ldots,n\}$ are denoted by $A^s_j$, with the index $s$ taking one of $k$ different values (with ``$s$'' standing for measurement \textit{s}etting); upon measurement, each observable realizes one out of its $l$ possible outcomes. If each party $j$ chooses to measure the observable $A^{s_j}_j$ on the system they have received, we describe this by a vector of settings $\vec{s}=s_1\ldots s_n$. The resulting statistics form a probability distribution over the $l^n$ possible joint outcomes of the observables
$$
A^{s_1}_1,\ldots, A^{s_n}_n .
$$
Upon varying $\vec{s}$, one obtains a family of joint distributions indexed by $\vec{s}$. The question asked by Bell was, which families---or \emph{correlations}---can arise in this way? He found that quantum theory allows correlations which cannot be obtained by classical probability theory and local interactions alone; there are correlations which defy \emph{local realism}. This phenomenon is now known as \emph{quantum nonlocality}.

The possibilities and limitations of quantum nonlocality have become a quite active field of research. As originally discovered by Tsirelson~\cite{Tsi2} and later popularized by Popescu and Rohrlich~\cite{PR}, there are correlations which respect causality in the sense of satisfying the so-called \emph{no-signaling principle}, but nevertheless cannot be realized in quantum theory, namely the \emph{Popescu--Rohrlich box}, or \emph{PR-box}. We refer to literature like~\cite{MAG,Sca} for further background on the no-signaling principle and the PR-box.

\subsection*{Main result.}

The contribution of this paper concerns the understanding of the set of local realistic correlations and the set of all no-signaling correlations. These are convex sets with a finite number of extreme points, known as the \emph{Bell polytope} and the \emph{no-signaling polytope}, respectively; for technical reasosns, we drop the normalization of probability and work with convex cones instead of convex polytopes. Our main result concerns $(n,2,2)$ scenarios, and identifies the no-signaling cone in such a scenario as the polyhedral dual of the Bell cone. This means that there is a linear bijective correspondence between inequalities valid for the Bell cone---the \emph{Bell inequalities}---and no-signaling boxes. In particular, this implies a bijective correspondence between facet Bell inequalities and extremal no-signaling boxes and explains why the authors of~\cite{PBS} found the same number of extremal no-signaling boxes in the $(3,2,2)$ scenario as Sliwa~\cite{Sliwa} found facet Bell inequalities. This coincidence has already been noted in~\cite{PBS,ycats}, where no explanation was found. Similar questions have been asked, and partially answered, in~\cite{HWB}.

Our main technical innovation is the use of convex cones. Dropping the normalization of probabilities leads to substantial simplifications as compared to working with convex polytopes. In comparison to the approach relying on tensor products of Banach spaces~\cite{JPPVW}, where one can treat the correlations only on the level of correlators while ignoring the marginal probabilities, using tensor products of convex cones naturally allows---in fact, necessitates---also the consideration of marginal probabilities. On the other hand, while the Banach space approach allows a mathematically natural description of quantum correlations in terms of tensor products of operator spaces~\cite{JPPVW}, we do not know at the moment whether and how this might be possible in our formalism; the theory of operator systems~\cite{CE} and their tensor products~\cite{KPTT} seems like a good candidate framework.

\subsection*{Summary}

We now outline the content of this paper in more detail. Firstly, the mathematical Appendices~\ref{tennot} and~\ref{conesapp} provide the necessary background material on tensor notation, convex cones, their tensor products, and duality. Since this does not seem to be standard material in nonlocality theory, we have included it here for convenient reference. The notation explained in Section~\ref{tennot} is standard in differential geometry. Section~\ref{conesapp} does not contain any original research either; rather, much of it is due to Namioka and Phelps~\cite{NP}. In Section~\ref{mainsec}, we start the main text with definitions and basic properties of the cone of local realistic correlations (the \emph{Bell cone}) $B_n^{2,2}$ and the cone of no-signaling correlations $NS_n^{2,2}$ (the \emph{no-signaling cone}) in $(n,2,2)$ scenarios. We identify them as $n$-fold tensor products of a certain $3$-dimensional cone $Sq$. While these observations seem to be folklore even more generally for all $(n,k,l)$ scenarios, we have not been able to find a detailed account anywhere else. The second part of Section~\ref{mainsec} states and proves the main theorem: the polyhedral dual of $B_n^{2,2}$ is $NS_n^{2,2}$. We also observe that this cannot hold in any other nontrivial Bell scenario $(n,k,l)$ with $(k,l)\neq (2,2)$. Section~\ref{apps} then continues by presenting some applications of our result. After discussing how to use symmetries in order to construct no-signaling boxes in the $(n+1,2,2)$ scenario from no-signaling boxes in the $(n,2,2)$ scenario, we apply our duality to transfer this to a construction of $(n+1,2,2)$-scenario Bell inequalities from $(n,2,2)$-scenario Bell inequalities. Besides the Mermin--Klyshko~\cite{GB} family of Bell inequalities, we also find that $11$ of the $46$ symmetry classes of facet Bell inequalities in the $(3,2,2)$ case~\cite{Sliwa} arise in this way. We also note that our duality theorem contains the classification of full-correlator Bell inequalities in the $(n,2,2)$ scenarios~\cite{WW,ZB} as a special case. The final Section~\ref{quantum} is devoted to the study of how the quantum set behaves under our duality. We argue that a reasonable hypothesis is to conjecture the quantum set to be self-dual. We find that every self-dual cone of correlations which respects the symmetries violates the CHSH Bell inequality~\cite{CHSH} precisely as much as quantum correlations do~\cite{Tsi2}. Finally, we use some recent results on quantum correlations in order to show that our self-duality hypothesis does not hold in the $(3,2,2)$ scenario.

\subsection*{Informal description of the duality}

We now describe how to transform a Bell inequality in the $(n,2,2)$ Bell scenario into a no-signaling box in the same scenario, and vice versa. This is an informal description of our main result Theorem~\ref{mainthm}. For a rigorous treatment and the proofs, see Section~\ref{mainsec}.

To begin, start with a Bell inequality written in correlator form and normalize it so that the right-hand side is $1$. For the sake of having a concrete example, we illustrate the procedure with inequality \#17 from~\cite{Sliwa}, which reads, in our notation and in normalized form,
\begin{align*}
\frac{1}{4}\langle A_1^1\rangle + \frac{1}{4}\langle A_1^2\rangle + \frac{1}{4}\langle A_1^1 A_2^1\rangle + \frac{1}{4}\langle A_1^2 A_2^1\rangle + \frac{1}{4}\langle A_1^1 A_3^1\rangle + \frac{1}{4}\langle A_1^2 A_3^1\rangle&\\
- \frac{1}{4}\langle A_1^1 A_2^1 A_3^1\rangle - \frac{1}{4}\langle A_1^2 A_2^1 A_3^1\rangle + \tfrac{1}{2}\langle A_1^1 A_2^2 A_3^2\rangle - \tfrac{1}{2}\langle A_1^2 A_2^2 A_3^2\rangle &\leq 1 \:.
\end{align*}
Now take this inequality and apply the replacements
\begin{align}
\begin{split}
\label{traf}
A_j^1&\longrightarrow A_j^1 + A_j^2 \\
A_j^2&\longrightarrow A_j^1 - A_j^2 
\end{split}
\end{align}
for every party $j=1,\ldots,n$. Upon using linearity of expectation values, this replaces each term by a sum of terms; for example,
$$
\langle A_1^1 A_2^2\rangle\longrightarrow \langle (A_1^1+A_1^2)(A_2^1-A_2^2)\rangle = \langle A_1^1A_2^1\rangle - \langle A_1^1A_2^2\rangle + \langle A_1^2 A_2^1\rangle - \langle A_1^2 A_2^2\rangle \:.
$$
Now collect all the terms and throw away the right-hand side, as the ensuing expression is typically not going to be a valid Bell inequality. In the example, we obtain by this procedure the expression
\begin{align}
\begin{split}
\label{boxex}
\tfrac{1}{2}\langle A_1^1\rangle + \tfrac{1}{2}\langle A_1^1 A_2^1\rangle + \tfrac{1}{2}\langle A_1^1 A_2^2\rangle + \tfrac{1}{2}\langle A_1^1 A_3^1\rangle + \tfrac{1}{2}\langle A_1^1 A_3^2\rangle - \tfrac{1}{2}\langle A_1^1 A_2^1 A_3^1\rangle - \tfrac{1}{2}\langle A_1^1 A_2^1 A_3^2\rangle \\
- \tfrac{1}{2}\langle A_1 A_2^2 A_3^1\rangle - \tfrac{1}{2}\langle A_1^1 A_2^2 A_3^2\rangle + \langle A_1^2 A_2^1 A_3^1\rangle - \langle A_1^2 A_2^1 A_3^2\rangle - \langle A_1^2 A_2^2 A_3^1\rangle + \langle A_1^2 A_2^2 A_3^2\rangle \:.
\end{split}
\end{align}
This expression defines a no-signaling box by taking the value of each correlator of the box to be given by its coefficient in the expression; if a certain correlator does not appear, then it should be regarded as having a coefficient of $0$. We recognize the no-signaling box defined by~(\ref{boxex}) as precisely no-signaling box \#40 of~\cite{PBS}.

Similarly, one can also start from a no-signaling box and transform it into a valid Bell inequality. Since this can be done precisely by following the same procedure backwards, we omit the details.

The map~(\ref{traf}) is a transformation which has been employed by Werner and Wolf~\cite{WW} and \.Zukowski and Brukner~\cite{ZB} in a very similar context. Therefore, the proof of our upcoming Theorem~\ref{mainthm} is essentially theirs.

In the following, we will make this informal description more precise and prove that it implements a duality between the Bell polytope and the no-signaling polytope. We use tensor notation (Appendix~\ref{tennot}) and the formalism of convex cones (Appendix~\ref{conesapp}), which we regard as the most adequate mathematical formalism for our present purposes.


\section{Polyhedral duality in Bell scenarios}
\label{mainsec}

\subsection{The Bell cone and the no-signaling cone.}

From here on, we write $A^{-1}_j$ and $A^{+1}_j$ for the two possible measurements of party $j=1,\ldots,n$, where both observables are taken to have outcomes in $\{-1,+1\}$. 

In order to allow for uniform notation which both covers full correlators like $\langle A^{\pm 1}_1\cdots A^{\pm 1}_n\rangle$ and single-party marginals like $\langle A^{\pm 1}_j\rangle$, we introduce a third setting ``$0$'' for each party which means ``do not measure at all''. Then each full or marginal correlator can be written as
$$
\langle A^{s_1}_1\cdots A^{s_n}_n\rangle
$$
for a certain $n$-tuple of settings $\vec{s}=\left(s_1,\ldots,s_n\right)\in\{-1,0,+1\}^n$. A full set of correlations is defined by specifying a number for this correlator for each $\vec{s}\in\{-1,0,+1\}^n$. We think of such a specification as a tensor $x^{\vec{s}}=x^{s_1\ldots s_n}\in\otimes^n\R^3$.


For $n$-tuples of settings $\vec{s}\in\{-1,+1\}^n$ and outcomes $\vec{t}\in\{-1,+1\}^n$, the outcome probability is given by, in terms of the correlators,
\beq
\label{nosig}
P(\vec{t}|\vec{s}) = g_{v_1\ldots v_n}(\vec{t},\vec{s}) \langle A_1^{v_1} \cdots A_n^{v_n} \rangle
\eeq
where the tensor $g_{v_1\ldots v_n}(\vec{t},\vec{s})=g_{\vec{v}}(\vec{t},\vec{s})$ is defined as
\beq
g_{\vec{v}}(\vec{t},\vec{s}) = 2^{-n} \prod_j \left\{ \begin{array}{ccl} t_j & \textrm{if} & v_j=s_j \\ 1 & \textrm{if} & v_j=0 \\ 0 & \textrm{if} & v_j=-s_j \end{array} \right.
\eeq
The \emph{no-signaling cone} $NS_n^{2,2}\subseteq\R^{3^n}=\otimes^n\R^3$ is defined by requiring all these probabilities to be nonnegative. In other words, a point $x^{\vec{s}}=x^{s_1\ldots s_n}\in\otimes^n\R^3$ lies in $NS_n^{2,2}$ if and only if
\beq
\label{nspos}
g_{\vec{v}}(\vec{t},\vec{s}) x^{\vec{v}} \geq 0 \quad\forall \vec{s},\vec{t}\in\{-1,+1\}^n,
\eeq
The component $x^{0\ldots 0}$ represents the value of the correlator $\langle A^0_1\ldots A^0_n\rangle$ with all settings equal to $0$, which is the normalization of probability; the inequalities~(\ref{nspos}) allow it to be any nonnegative number. The usual no-signaling polytope, as introduced in~\cite{BLMPPR}, can be obtained by intersecting $NS_n^{2,2}$ with the hyperplane $x^{0\ldots 0}=1$.

We now discuss the analogous structure for those correlations compatible with local realistic models. A local deterministic strategy is defined by specifying $n$-tuples of signs $\vec{a}^{-1}\in\{-1,+1\}^n$ and $\vec{a}^{+1}\in\{-1,+1\}^n$ such that the correlators of this strategy are given by, with $a^0_j=1$ for all $j$,
$$
\langle A^{s_1}_1\ldots A^{s_n}_n\rangle = a^{s_1}_1\cdots a^{s_n}_n ,
$$
which is again a tensor in $\R^{3^n}=\otimes^n\R^3$. By definition, the \emph{Bell cone} $B_n^{2,2}\subseteq\R^{3^n}$ is
\beq
\label{defB}
B_n^{2,2} = \cone\left(\left\{ a^{s_1}_1\cdots a^{s_n}_n \::\: \vec{a}^{\pm 1}\in\{-1,+1\}^n, \: \vec{a}^0 = \vec{1} \right\}\right)
\eeq
Just as in the no-signaling case, one obtains the usual Bell polytope by intersecting $B_n^{2,2}$ with the hyperplane $x^{0\ldots 0}=1$.

\begin{lem}
For $n=1$, the no-signaling cone and the Bell cone coincide, and are given by the \emph{square}
\beq
\label{Sqdef}
Sq = \left\{ x^s \in\R^3 \:\bigg|\: -x^0\leq x^{-1}\leq x^0 ,\: -x^0 \leq x^{+1}\leq x^0 \right\} .
\eeq
\end{lem}

Equivalently, $Sq$ can be defined as
\beq
\label{Sqex}
Sq = \cone\left\{ \left(\begin{array}{c}-1\\ +1\\ -1\end{array}\right), \left(\begin{array}{c}-1\\ +1\\ +1\end{array}\right), \left(\begin{array}{c}+1\\ +1\\ -1\end{array}\right), \left(\begin{array}{c}+1\\ +1\\ +1\end{array}\right) \right\} ,
\eeq
This cone is a square in the sense that its normalized part $\{x^s\in Sq \:|\: x^0=1\}$ is a square; see Figure~\ref{sqfig}.

\begin{proof}
For $n=1$, the inequalities~(\ref{nspos}) defining the no-signaling cone are given by
$$
x_0 - x_{-1} \geq 0, \quad x_0 + x_{-1} \geq 0, \quad x_0 - x_{+1} \geq 0,\quad x_0 + x_{+1} \geq 0,
$$
which are precisely the defining inequalities of $Sq$. On the other hand, evaluating the definition~(\ref{defB}) of the Bell cone for $n=1$ gives precisely~(\ref{Sqex}).
\end{proof}

While we have not been able to find a formal proof in the literature, it seems to be folklore knowledge that Bell polytopes can be identified with minimal tensor products, while no-signaling polytopes can be thought of as maximal tensor products. We now make this precise in our framework.

\begin{lem}
\begin{enumerate}
\item The Bell cone $B_n^{2,2}$ equals the minimal tensor product
\beq
\label{Bmin}
B_n^{2,2} = \underbrace{Sq\otimes_{\min}\ldots\otimes_{\min} Sq}_{n\textrm{ factors}} .
\eeq
\item The no-signaling cone $NS_n^{2,2}$ equals the maximal tensor product
\beq
\label{NSmax}
NS_n^{2,2} = \underbrace{Sq\otimes_{\max}\ldots\otimes_{\max} Sq}_{n\textrm{ factors}} .
\eeq
\end{enumerate}
\end{lem}

\begin{proof}
\begin{enumerate}
\item
By the definition~(\ref{defB}), every local deterministic point $x^{\vec{s}}=a_1^{s_1}\cdots a_n^{s_n}$ is a tensor product of points $a_j^s\in Sq$. The assertion follows since both sides of~(\ref{Bmin}) are defined as the cone generated by these tensor products.
\item This is the case since inequality~(\ref{nspos}) can be rewritten as
$$
h_{v_1}(t_1,s_1)\cdots h_{v_n}(t_n,s_n) x^{v_1\ldots v_n} \geq 0
$$
with
$$
h_v(t,s) = \left\{ \begin{array}{ccl} t & \textrm{if} & v=s \\ 1 & \textrm{if} & v=0 \\ 0 & \textrm{if} & v=-s \end{array} \right.
$$
\end{enumerate}
\end{proof}

\subsection{Bell inequalities.} By definition, a \emph{Bell inequality} in the $(n,2,2)$ scenario is an element of the dual cone $(B_n^{2,2})^*$, i.e.~a tensor $f_{\vec{s}}=f_{s_1\ldots s_n}$ such that
$$
f_{s_1\ldots s_n}x^{s_1\ldots s_n}\geq 0
$$
for all $x^{\vec{s}}\in B_n^{2,2}$. The simplest examples are the trivial Bell inequalities~(\ref{nspos}), which express the nonnegativity of outcome probabilities. A Bell inequality is a \emph{facet} if it is an extreme ray of $(B_n^{2,2})^*$.

\subsection{Main Theorem.} After formalizing these preliminaries, we are now almost in a position to state and prove our main duality theorem. 

In differential geometry and general relativity~\cite{Einstein}, the metric tensor $g_{st}$ allows a conversion from vectors $x^s\in V$ to covectors $g_{st}x^t\in V^*$ (``lowering the index'') and conversely from covectors $x_s\in V^*$ to vectors $g^{st}x_t\in V$ (``raising the index''). The same r\^ole is played in the Hilbert space formalism of quantum mechanics by the scalar product, which defines the (antilinear) map assigning to each ket $|\psi\rangle$ its associated bra $\langle\psi|$. We now introduce a similar tensor in our framework, which lets us identify $Sq$ with $Sq^*$. It is the matrix representation of~(\ref{traf}).

\begin{lem}
\begin{enumerate}
\item The two tensors
$$
F^{st} = \left( \begin{array}{ccc} 1 & 0 & 1\\ 0 & 1 & 0\\ 1 & 0 & -1 \end{array}\right) \qquad F_{st} = \left( \begin{array}{ccc} \tfrac{1}{2} & 0 & \tfrac{1}{2}\\ 0 & 1 & 0\\ \tfrac{1}{2} & 0 & -\tfrac{1}{2} \end{array}\right)
$$
are inverse to each other: $F_{st}F^{tv}=\delta_s^v$ and $F^{st}F_{tv}=\delta_v^s$.
\item 
The linear maps
\beq
\label{F}
Sq\to Sq^*, \quad x^s \mapsto F_{st}x^t \qquad Sq^*\to Sq,\quad x_s\mapsto F^{st}x_t
\eeq
are mutually inverse and implement an isomorphism $Sq\stackrel{\cong}{\longleftrightarrow}Sq^*$ (see Figure~\ref{sqfig}).
\end{enumerate}
\end{lem}

\begin{figure}
\begin{center}
\begin{tikzpicture}
\draw[->] (-2,0) -- (2,0) node[anchor=north west] {$x^{-1}$};
\draw[->] (0,-2) -- (0,2) node[anchor=south east] {$x^{+1}$};
\draw[thick] (-1.2,-1.2) -- (-1.2,1.2) -- (1.2,1.2) -- (1.2,-1.2) -- (-1.2,-1.2);
\node[anchor=north west] at (1.2,0) {$x^0$};
\node[anchor=south east] at (0,1.2) {$x^0$};
\node at (-2,2) {$Sq$};
\draw (1.2,-.2) -- (1.2,.2);
\draw (-.2,1.2) -- (.2,1.2);
\draw[<->] (2,.5) .. controls (3.5,1) .. (5,.5);
\node[anchor=south] at (3.5,.8) {(\ref{F})};
\draw[->] (5,0) -- (9,0) node[anchor=north west] {$x_{-1}$};
\draw[->] (7,-2) -- (7,2) node[anchor=south east] {$x_{+1}$};
\draw[thick] (7,-1.2) -- (5.8,0) -- (7,1.2) -- (8.2,0) -- (7,-1.2);
\node[anchor=north west] at (8.2,0) {$x_0$};
\node[anchor=south east] at (7,1.2) {$x_0$};
\node at (5,2) {$Sq^*$};
\end{tikzpicture}
\end{center}
\caption{Sections of the cone $Sq$ and its dual $Sq^*$ at a fixed value of $x_0$ respectively $x^0$. The isomorphism~(\ref{F}) is a linear bijection mapping $Sq$ to $Sq^*$.}
\label{sqfig}
\end{figure}
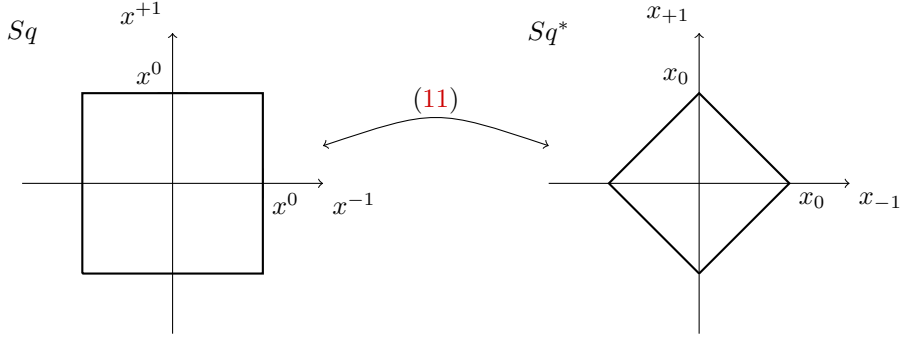

\begin{proof}
\begin{enumerate}
\item Matrix multiplication.
\item By~(\ref{coneS2}) and~(\ref{Sqdef}), we have
\beq
\label{Sqdual}
Sq^* = \cone\left\{ (-1,+1,0),\: (0,+1,-1),\: (+1,+1,0),\: (0,+1,+1)\right\} ,
\eeq
and a direct calculation shows that~(\ref{F}) implements a bijection between the extreme rays of $Sq$ and those of $Sq^*$.
\end{enumerate}
\end{proof}

Therefore, also the $n$-fold tensor product of~(\ref{F}), which is the pair of mutually inverse linear maps
$$
x^{s_1\ldots s_n} \mapsto F_{s_1t_t}\cdots F_{s_nt_n} x^{t_1\ldots t_n} ,\qquad x_{s_1\ldots s_n} \mapsto F^{s_1t_t}\cdots F^{s_nt_n} x_{t_1\ldots t_n}
$$
implement an isomorphism
\beq
\label{Fn}
Sq\omax\ldots\omax Sq\stackrel{\cong}{\longleftrightarrow} Sq^*\omax\ldots\omax Sq^* .
\eeq

\begin{thm}
\label{mainthm}
This implements a linear isomorphism $NS_n^{2,2}\stackrel{\cong}{\longleftrightarrow}(B_n^{2,2})^*$.
\end{thm}

In terms of the machinery we have built up, the proof is very simple and intuitive.

\begin{proof}
\begin{align*}
(B^{2,2}_n)^* &\stackrel{(\ref{Bmin})}{=} \left( Sq \otimes_{\mathrm{min}} \ldots \otimes_{\mathrm{min}} Sq \right)^* \\
 &\stackrel{(\ref{minmaxdual})}{=} Sq^* \otimes_{\mathrm{max}}  \ldots \otimes_{\mathrm{max}} Sq^* \\
 &\stackrel{(\ref{Fn})}{\cong} Sq \otimes_{\mathrm{max}}  \ldots \otimes_{\mathrm{max}} Sq \stackrel{(\ref{NSmax})}{=} NS^{2,2}_n . 
\end{align*}
\end{proof}

In particular, it follows that we have a bijective correspondence between (facet) Bell inequalities and (extremal) no-signaling boxes. The proof of the Theorem shows that this correspondence is precisely the one described in the Introduction. Before developing some more general theory, we now turn to examples.

For $n=2$, it is instructive to verify the Theorem explicitly. Since we have already given an example of the transformation from Bell inequalities to no-signaling boxes in the introduction, we now go the other way. With $n=2$, i.e.~in the CHSH scenario~\cite{CHSH}, there are two symmetry classes of extremal no-signaling boxes: the local deterministic boxes and the PR-boxes~\cite{PR,Sca,Tsi3}. As a representative of the first class, we consider the box with all observables having a deterministic $+1$ outcome; in our framework, this box is given by $\langle A_1^{s_1} A_2^{s_2}\rangle =1$ for all $\vec{s}=(s_1,s_2)$. By Theorem~\ref{mainthm}, the associated Bell inequality $f_{\vec{s}}x^{\vec{s}}\geq 0$ has coefficients
$$
f_{s_1 s_2} = F_{s_1 t_1} F_{s_2 t_2} \langle A_1^{t_1} A_2^{t_2}\rangle = \sum_{t_1,t_2} F_{s_1 t_1} F_{s_2 t_2} = \left(\sum_{t_1} F_{s_1 t_1}\right) \left(\sum_{t_2} F_{s_2 t_2}\right) 
$$
Since
$$
\sum_t F_{st} = \left\{\begin{array}{ccc} 1 & \textrm{for} & s=-1 \\ 1 & \textrm{for} & s=0 \\ 0 & \textrm{for} & s=+1 \end{array}\right. ,
$$
the resulting Bell inequality is
$$
x^{-1,-1}+x^{-1,0}+x^{0,-1}+x^{0,0}\geq 0 ,
$$
which is the nonnegativity of probability inequality~(\ref{nspos}) with $\vec{s}=(-1,-1)$ and $\vec{t}=(+1,+1)$. On the other hand, the other symmetry class of extremal boxes is the class of the PR-box
\beq
\langle A_1^{s_1} A_2^{s_2} \rangle = \begin{array}{c|ccc} s_1,\: s_2 & -1 & 0 & +1 \\\hline -1 & 1 & 0 & 1 \\ 0 & 0 & 1 & 0 \\ -1 & 1 & 0 & -1 \end{array} .
\eeq
The PR-box is very special in that it satisfies the simple relation $\langle A_1^{s_1} A_2^{s_2}\rangle = F^{s_1 s_2}$. Using this, the coefficients of the resulting Bell inequality can easily be calculated as
$$
f_{s_1 s_2} = F_{s_1 t_1} F_{s_2 t_2} \langle A_1^{t_1} A_2^{t_2}\rangle = F_{s_1 t_1} F_{s_2 t_2} F^{t_1 t_2} = \delta_{s_1}^{t_2} F_{s_2 t_2} = F_{s_1 s_2} .
$$
These are precisely the coefficients of the CHSH inequality
$$
\tfrac{1}{2} x^{-1,-1} + \tfrac{1}{2} x^{-1,+1} + \tfrac{1}{2} x^{+1,-1} - \tfrac{1}{2} x^{+1,+1} + x^{0,0} \geq 0 \:.
$$
This ends our discussion of the $n=2$ case.

For $n=3$, a case for which also both the facet Bell inequalities~\cite{Sliwa} and the extremal no-signaling boxes~\cite{PBS} are known, Theorem~\ref{mainthm} can be nicely verified by a computer calculation. In terms of the numberings of~\cite{PBS,Sliwa}, this has resulted in Table~\ref{correspond}.

We now continue the general theory.

\begin{table}[t]
\begin{tabular}{c|cccccccccccccccc}
Tight Bell inequality \# & 1 & 2  & 3  & 4  & 5  & 6  & 7  & 8  & 9 & 10 & 11 & 12 & 13 & 14 & 15 & 16 \\\hline
Extremal no-signaling box \# & 1 & 46 & 45 & 2 & 29 & 25 & 44 & 41 & 42 & 4 & 7 & 5 & 6 & 3 & 8 & 27 
\end{tabular}\\[.5cm]
\begin{tabular}{c|cccccccccccccccc}
Tight Bell inequality \# & 17 & 18 & 19 &20 & 21 & 22 & 23 & 24 & 25 & 26 & 27 & 28 & 29 &30 & 31 & 32 \\\hline 
Extremal no-signaling box \# & 40 & 11 & 12 & 10 & 19 & 21 & 9 & 30 & 32 & 31 & 33 & 39 & 38 & 17 & 35 & 14
\end{tabular}\\[.5cm]
\begin{tabular}{c|cccccccccccccccc}
Tight Bell inequality \# & 33 & 34 & 35 & 36 & 37 & 38 & 39 &40 & 41 & 42 & 43 & 44 & 45 & 46 \\\hline
Extremal no-signaling box \# & 34 & 13 & 16 & 26 & 28 & 24 & 22 & 23 & 43 & 36 & 37 & 18 & 15 & 20
\end{tabular}\\[.5cm]
\caption{Correspondence between the facet Bell inequalities as classified in~\cite{Sliwa} and the extremal no-signaling boxes as classified in~\cite{PBS}. (Courtesy of Jean-Daniel Bancal.)}
\label{correspond}
\end{table}

\begin{cor}
\label{cordual}
The correlations $x^{\vec{s}}$ lie in $B_n^{2,2}$ (resp.~$NS_n^{2,2}$) if and only if
\beq
\label{Fxy}
F_{s_1t_1}\cdots F_{s_nt_n} x^{s_1\ldots s_n} y^{t_1\ldots t_n} \geq 0
\eeq
for all $y^{\vec{s}}$ in $NS_n^{2,2}$ (resp.~$B_n^{2,2}$).
\end{cor}

\begin{proof}
Theorem~\ref{mainthm} and Lemma~\ref{basic}.\ref{basicdual} together imply the statement in the brackets. Lemma~\ref{basic}.\ref{doubledual} then shows that $y^{\vec{s}}$ lies in $B_n^{2,2}$ if and only if~(\ref{Fxy}) holds for all $x^{\vec{s}}\in NS_n^{2,2}$. Renaming $x^{\vec{s}}\leftrightarrow y^{\vec{s}}$ proves that $x^{\vec{s}}$ lies in $B_n^{2,2}$ if and only if
$$
F_{t_1s_1}\cdots F_{t_ns_n} x^{s_1\ldots s_n} y^{t_1\ldots t_n} \geq 0
$$
for all $y^{\vec{s}}$ in $NS_n^{2,2}$. This proves the claim since $F_{st}=F_{ts}$.
\end{proof}

As far as we can see, there is little relation between a (facet) Bell inequality and its associated (extremal) no-signaling box. For example, again in the numbering of~\cite{PBS,Sliwa}, inequality \#10 corresponds to box \#4, although inequality \#10 is only violated by boxes \#25 and \#29~\cite{PBS}. The only obvious statement along these lines is that local realistic no-signaling boxes, i.e.~those lying in $B_n^{2,2}$, correspond to trivial Bell inequalities, i.e.~those lying in $(NS_n^{2,2})^*$, not violated by any no-signaling boxes:

\begin{prop}
The correlations $x^{\vec{s}}$ lie in $B_n^{2,2}$ if and only if $F_{s_1t_1}\cdots F_{s_nt_n}x^{t_1\ldots t_n}\in (NS_n^{2,2})^*$.
\end{prop}

\begin{proof}
Corollary~\ref{cordual} together with Lemma~\ref{basic}.\ref{basicdual}.
\end{proof}

All our results so far concern the case of Bell scenarios with two binary observables per party. Could it be that Theorem~\ref{mainthm} holds in greater generality?

\begin{prop}
For $(k,l)\neq (2,2)$ and any $n\in\N$, there is no polyhedral duality in the $(n,k,l)$ Bell scenario between no-signaling boxes and Bell inequalities.
\end{prop}

\begin{proof}
If the duality would exist, then in particular the number of local deterministic points of the Bell cone would coincide with the number of facets of the no-signaling cone. It is well-known that the former is $l^{kn}$, while the latter is $(lk)^n$. Since $k,l\geq 2$ for any nontrivial Bell scenario, the equation $l^{kn}=(lk)^n$, or $l^{k-1}=k$, implies $l=k=2$.
\end{proof}

\section{Applications}
\label{apps}

\subsection{Extending to more parties.}

What might Theorem~\ref{mainthm} be useful for? One thing one could ask is this: given the several known families of Bell inequalities in $(n,2,2)$ scenarios (e.g.~Mermin--Klyshko~\cite{GB}, generalized Svetlichny~\cite{gSv}), what are the families of no-signaling boxes associated to these inequalities? However, since it is unclear what the significance of these boxes would be, we have not done this yet. Instead, we have noticed that the terms appearing in the iterative definition of the Mermin--Klyshko inequalities (e.g.~\cite[eq.~(4)]{GB}) are precisely the ones which appear in the transformation~\ref{traf}. With the Mermin--Klyshko term $\mathcal{M}_n$ defined inductively by $\mathcal{M}_0=\mathcal{M}'_0=1$ and
\begin{align}
\begin{split}
\label{MK2}
\mathcal{M}_{n+1} &= \mathcal{M}_n \cdot \tfrac{1}{2}\left(A^{-1}_{n+1} + A^{+1}_{n+1}\right) + \mathcal{M}'_n \cdot \tfrac{1}{2} \left(A^{-1}_{n+1} - A^{+1}_{n+1}\right) \\
\mathcal{M}'_{n+1} &= \mathcal{M}'_n\cdot \tfrac{1}{2}\left(A^{-1}_{n+1} + A^{+1}_{n+1} \right) - \mathcal{M}_n \cdot \tfrac{1}{2} \left( A^{-1}_{n+1} - A^{+1}_{n+1}\right) ,
\end{split}
\end{align}
the $n$-party Mermin--Klyshko inequality is given by, in our notation,
\beq
\label{MK1}
\langle \mathcal{M}_{n}\rangle - \langle A_1^0\ldots A_n^0 \rangle \geq 0 .
\eeq
We have tried to relate the appearance of the familiar terms in~(\ref{MK2}) to our tensor $F_{st}$. These considerations have led us to a general construction of extending Bell inequalities with two binary observables from $n$ parties to $n+1$ parties. The Mermin--Klyshko family has turned out to be a particular instance of this construction; see the upcoming Example~\ref{MKex}.

We begin the presentation of our construction by first discussing how to extend no-signaling boxes with two binary observables from $n$ parties to $n+1$ parties, and later applying Theorem~\ref{mainthm} to transfer this construction to Bell inequalities. This requires some preparation.

\begin{lem}
\label{extendboxes}
The correlations $x^{s_1\ldots s_{n+1}}$ lie in $NS_{n+1}^{2,2}$ if and only if
\beq
\label{extend}
x^{s_1\ldots s_{n}0} \pm x^{s_1\ldots s_{n}t}
\eeq
for both signs and both $t\in\{-1,+1\}$.
\end{lem}

\begin{proof}
This follows from an application of Lemma~\ref{maxcrit} to~(\ref{NSmax}), since each of the four possible forms of~(\ref{extend}) is equal to $h_v x^{s_1\ldots s_n v}$ for $h_v\in\ex(Sq^*)$ one of the four extreme rays from~(\ref{Sqdual}).
\end{proof}

The no-signaling cone $NS_n^{2,2}$ has many symmetries. In particular, flipping the outcome of some of the observables, i.e.~applying any number of transformations of the form $A^{s}_j\to -A^s_j$, is a symmetry in the sense that it is a linear isomorphism mapping mapping $NS_n^{2,2}$ onto itself.

More specifically, flipping the outcomes of some of the observables is a symmetry which is inverse to itself; in other words, applying it twice gives the identity map. Such a symmetry is an \emph{involution}, by which we mean a linear isomorphism
$$
\iota:NS_n^{2,2}\stackrel{\cong}{\longrightarrow} NS_n^{2,2}
$$
satisfying $\iota^2=\id$.

Like every involution on a vector space, $\iota$ splits $\otimes^n\R^{3}$ into the direct sum of an \emph{odd} subspace (eigenvalue $-1$) and an \emph{even} subspace (eigenvalue $+1$). 

\begin{prop}
\label{construction}
Let $x^{\vec{s}},y^{\vec{s}}\in\otimes^n\R^{3}$ be such that $\iota(x^{\vec{s}})=x^{\vec{s}}$ and
\beq
\label{alphabeta}
(x^{\vec{s}}\pm y^{\vec{s}})\in NS_n^{2,2}
\eeq
for both choices of sign. Then the correlations defined by
$$
z^{s_1\ldots s_{n+1}} = \left\{
 \begin{array}{ccl} y^{\vec{s}} & \mathrm{if} & s_{n+1}=-1 \\ x^{\vec{s}} & \mathrm{if} & s_{n+1} = 0 \\ \iota(y^{\vec{s}}) & \mathrm{if} & s_{n+1}=+1 \end{array} \right.
$$
are in $NS_{n+1}^{2,2}$.
\end{prop}

\begin{proof}
This follows directly from the assumptions upon an application of Lemma~\ref{extendboxes}.
\end{proof}

Recognizing whether a given $(n+1)$-partite no-signaling box $z^{s_1\ldots s_{n+1}}$ can be constructed in this way for a given $\iota$ is not difficult: one only needs to check whether $z^{\vec{s}0}\in\otimes^n\R^{3}$ is even with respect to $\iota$, and whether $\iota$ maps $z^{\vec{s},-1}\in\otimes^n\R^{3}$ to $z^{\vec{s},+1}\in\otimes^n\R^{3}$.

We have found that 16 out of the 46 classes of extremal boxes for $n=3$ are of this form for some appropriate involution $\iota$; see Table~\ref{canboxes}.

However, Proposition~\ref{construction} is not very useful yet, since it is unclear when condition~(\ref{alphabeta}) is satisfied. In order to remedy this situation, it is possible to specialize further to those $x^{\vec{s}}$ and $y^{\vec{s}}$ which can be constructed from a \emph{single} $w^{\vec{s}}\in NS_n^{2,2}$ in such a way that~(\ref{alphabeta}) is automatic. To this end, we assume that we have another involution
$$
\kappa : NS_n^{2,2} \longrightarrow NS_n^{2,2}
$$
at our disposal, which commutes with the previous one, $\iota\kappa=\kappa\iota$. (In many cases, this is automatic: for example, if $\kappa$ is the involution which flips the outcomes of all observables, then $\kappa$ commutes with all other involutions generated by flips of outcomes, relabelings of observables, and permutations of parties.)

Due to their commutativity, $\iota$ and $\kappa$ are simultaneously diagonalizable, so that they define a decomposition of $\otimes^n\R^{3}$ into a direct sum of four linear subspaces indexed by the eigenvalue pairs $(\lambda_\iota,\lambda_\kappa)$ with $\lambda_\iota,\lambda_\kappa\in\{-1,+1\}$. The upcoming equation~(\ref{noeigen}) states that the component of the given $w^{\vec{s}}$ in the eigenspace associated to $(\lambda_\iota,\lambda_\kappa)=(-1,+1)$ vanishes.

\begin{prop}
\label{construction2}
Let $w^{\vec{s}}\in NS_n^{2,2}$ be such that
\beq
\label{noeigen}
\iota\kappa(w^{\vec{s}}) = \kappa(w^{\vec{s}}) - \iota(w^{\vec{s}}) + w^{\vec{s}} .
\eeq
Then the correlations defined by
\beq
z^{s_1\ldots s_{n+1}} = \left\{
 \begin{array}{ccl} w^{\vec{s}} - \kappa(w^{\vec{s}}) & \mathrm{if} & s_{n+1}=-1 \\ w^{\vec{s}} + \kappa(w^{\vec{s}}) & \mathrm{if} & s_{n+1} = 0 \\ \iota(w^{\vec{s}}) - \iota\kappa(w^{\vec{s}}) & \mathrm{if} & s_{n+1}=+1 \end{array} \right.
\eeq
are in $NS_{n+1}^{2,2}$.
\end{prop}

\begin{proof}
Apply~\ref{construction} with $y^{\vec{s}}=w^{\vec{s}}-\kappa(w^{\vec{s}})$ and $x^{\vec{s}}=w^{\vec{s}}+\kappa(w^{\vec{s}})$, since then the condition $x^{\vec{s}}\pm y^{\vec{s}}\in NS_n^{2,2}$ is automatic, and $\iota(x^{\vec{s}})=x^{\vec{s}}$ is guaranteed by~(\ref{noeigen}).
\end{proof}

As before, it is simple to check whether a given no-signaling box can be constructed in this way for a given $\iota$ and $\kappa$ as follows: $z^{\vec{s}0}$ needs to be even under $\kappa$, while $z^{\vec{s},-1}$ and $z^{\vec{s},+1}$ should both be odd under $\kappa$; also, $\iota$ should map $z^{\vec{s},-1}$ to $z^{\vec{s},+1}$; if all these conditions hold, then $\vec{w}$ can be constructed as $\tfrac{1}{2}(z^{\vec{s}0}+z^{\vec{s},-1})$, and~(\ref{noeigen}) is automatic.

For $n=3$, we have displayed those extremal no-signaling which can be constructed in this way from $n=2$, together with the corresponding involutions, in Table~\ref{canboxes} as those entries for which $\kappa$ is listed.

\newcommand{\llra}{\longleftrightarrow}
\renewcommand*\arraystretch{3.0}

\begin{table}
\begin{tabular}{cccc}
\phantom{xx}Box\phantom{xx} & \phantom{xx}New party\phantom{xx} & $\iota$ & $\kappa$ \\
\hline
$1$ & $A$ & $\id$ & $\id$ \\\hline
$2$ & $B$ & $C_1\llra -C_1$ & \pbox{20cm}{$C_0\llra -C_0$, $C_1\llra -C_1$} \\\hline
$3$ & $B$ & $C_1\llra -C_1$ & \pbox{20cm}{$A_0\llra -A_0$, $C_1\llra -C_1$} \\\hline
$4$ & $A$ & \pbox{20cm}{$B_0\llra C_0$, $B_1\llra -C_1$} & \pbox{20cm}{$B_0\llra -B_0$, $C_0\llra -C_0$} \\\hline
$6$ & $C$ & $B_1\llra -B_1$ & \pbox{20cm}{$A_1\llra -A_1$} \\\hline
$7$ & $A$ & \pbox{20cm}{$B_0\llra B_1$, $C_1\llra -C_1$} & \pbox{20cm}{$B_0\llra -B_0$, $B_1\llra -B_1$} \\\hline
$8$ & $B$ & $A_1\llra -A_1$ & \pbox{20cm}{$A_0\llra -A_0$, $C_0\llra -C_0$, $C_1\llra -C_1$} \\\hline
$9$ & $C$ & $B_1\llra -B_1$ & \pbox{20cm}{$A_1\llra -A_1$, $B_1\llra -B_1$} \\\hline
$10$ & $C$ & $B_1\llra -B_1$ & \pbox{20cm}{$A_1\llra -A_1$, $B_0\llra -B_0$, $B_1\llra -B_1$} \\\hline
$11$ & $C$ & \pbox{20cm}{$B_0\llra B_1$, $A_1\llra -A_1$} & --- \\\hline
$12$ & $C$ & $A_1\llra -A_1$ & --- \\\hline
$24$ & $B$ & $C_1\llra -C_1$ & --- \\\hline
$25$ & $A$ & \pbox{20cm}{$B_0\llra C_0$, $B_1\llra C_1$} & --- \\\hline
$31$ & $A$ & \pbox{20cm}{$B_0\llra B_1$, $C_0\llra C_1$} & --- \\\hline
$45$ & $C$ & $A_1\llra -A_1$ & \pbox{20cm}{$A_0\llra -A_0$, $A_1\llra -A_1$} \\\hline
$46$ & $A$ & \pbox{20cm}{$B_1\llra -B_1$, $C_1\llra -C_1$} & \pbox{20cm}{$A_0\llra -A_0$, $A_1\llra -A_1$}\\[.5cm]
\end{tabular}
\caption{List of extremal $3$-partite no-signaling boxes in the numbering of~\cite{PBS} which can be constructed as in Proposition~\ref{construction}, together with the required involution $\iota$ defined by specifying the observables on which it acts nontrivially, in the notation of~\cite{PBS}. Those for which $\kappa$ is listed also arise from the more convenient construction of Proposition~\ref{construction2}.}
\label{canboxes}
\end{table}

\renewcommand*\arraystretch{1}

Applying Theorem~\ref{mainthm} directly translates Proposition~\ref{construction2} into a method to extend Bell inequalities from $n$ parties to $n+1$ parties. In the following, we abuse notation a bit by considering $\iota$ and $\kappa$ as involutions acting on covariant tensors, or equivalently, on Bell inequalities as
$$
\iota,\kappa:(B_n^{2,2})^*\longrightarrow (B_n^{2,2})^*.
$$
This should not lead to confusion since, for example, the operation of relabeling the observables of some party is very much the same on Bell inequalities as on no-signaling boxes.

\begin{cor}
If 
\beq
\label{belln}
f_{s_1\ldots s_n} \langle A_1^{s_1}\ldots A_n^{s_n} \rangle \geq 0
\eeq
is a Bell inequality satisfying $\iota\kappa(f_{\vec{s}}) = \iota(f_{\vec{s}}) - \kappa(f_{\vec{s}}) + f_{\vec{s}}$, then 
\begin{align*}
 \tfrac{1}{2}\left[f_{s_1\ldots s_n} - \kappa(f_{s_1\ldots s_n}) - \iota(f_{s_1\ldots s_n}) + \kappa\iota(f_{s_1\ldots s_n}) \right] \langle & A_1^{s_1}\ldots A_n^{s_n}A_{n+1}^{-1}\rangle \\
 + \left( f_{s_1\ldots s_n} + \kappa(f_{s_1\ldots s_n}) \right) \langle & A_1^{s_1}\ldots A_n^{s_n}A_{n+1}^0 \rangle \\
+  \tfrac{1}{2}\left[ f_{s_1\ldots s_n} - \kappa(f_{s_1\ldots s_n}) + \iota(f_{s_1\ldots s_n}) - \kappa\iota(f_{s_1\ldots s_n}) \right] \langle & A_1^{s_1}\ldots A_n^{s_n}A_{n+1}^{+1}\rangle  \geq 0
\end{align*}
is an $(n+1)$-partite Bell inequality.
\end{cor}

In words, the procedure is as follows. One needs to start with any $n$-partite Bell inequality~(\ref{belln}). One decomposes this Bell inequality into a sum of four parts, where each part is an eigenvector for both $\iota$ and $\kappa$; the whole procedure works as long as the part associated to the pair of eigenvalues $(\lambda_\iota,\lambda_\kappa)=(-1,+1)$ vanishes, so that there are at most three nontrivial parts. Now in the part which is odd under $\kappa$ and odd under $\iota$, one adds the observable $A_{-1}^{n+1}$ in each term; in the part which is odd under $\kappa$ but even under $\iota$, one adds $A_{+1}^{n+1}$ instead. The terms even under $\kappa$ stay unchanged: these terms are not full $n$-party correlators after the extension to $n+1$ parties.

Due to the duality, it follows from Table~\ref{canboxes} that 11 of the 46 facet Bell inequalities in the $(3,2,2)$ scenario arise in this way from Bell inequalities in the $(2,2,2)$ scenario.

\begin{expl}[CHSH]
As a basic example, we show to construct the CHSH inequality~\cite{CHSH}. To this end, we start from
$$
\langle A_1^{-1} \rangle + \langle A_1^0 \rangle \geq 0
$$
which is a $1$-party ``Bell inequality''. With respect to the involutions $\iota:A_1^{-1}\leftrightarrow A_1^{+1}$ and $\kappa:A_1^s\leftrightarrow -A_1^s$, this inequality decomposes into the terms
$$
\underbrace{\tfrac{1}{2}A_{-1}-\tfrac{1}{2}A_{+1}}_{(\lambda_\iota,\lambda_\kappa)=(-1,-1)} 
+ \underbrace{\tfrac{1}{2}A_{-1}+\tfrac{1}{2}A_{+1}}_{(\lambda_\iota,\lambda_\kappa)=(+1,-1)}
+ \underbrace{A_0}_{(\lambda_\iota,\lambda_\kappa)=(+1,+1)} \geq 0
$$
so that our construction produces the new inequality
$$
\tfrac{1}{2}\langle A_1^{-1}A_2^{-1}\rangle - \tfrac{1}{2} \langle A_1^{+1}A_2^{-1} \rangle + \tfrac{1}{2} \langle A_1^{-1}A_2^{+1}\rangle + \tfrac{1}{2}\langle A_1^{+1}A_2^{+1}\rangle + \langle A_1^0 A_2^0\rangle \geq 0
$$
which is one version of the CHSH inequality in our notation.
\end{expl}

\begin{expl}[Mermin--Klyshko inequalities]
\label{MKex}
The previous example can be extended to cover the recursive definition of the Mermin--Klyshko inequalities as in~(\ref{MK2}),~(\ref{MK1}). For $n$ parties, we define $\iota$ to be the involution which exchanges $A_j^{-1}\leftrightarrow A_j^{+1}$ for all parties $j$; and $\kappa$ to be the involution flipping $A_n^s\leftrightarrow -A_n^s$. Then it follows from~(\ref{MK2}) by induction that $\iota$ interchanges $\mathcal{M}_n$ with $\mathcal{M}'_n$. Moreover, it also follows that $\mathcal{M}_n$ is odd under $\kappa$, while obviously the second term of~(\ref{MK1}) is even; upon also considering the action of $\iota$, one finds that the $n$-partite Mermin--Klyshko inequality decomposes as
$$
\underbrace{\tfrac{1}{2}\langle\mathcal{M}_n-\mathcal{M}'_n\rangle}_{(\lambda_\iota,\lambda_\kappa)=(-1,-1)} 
+ \underbrace{\tfrac{1}{2}\langle\mathcal{M}_n+\mathcal{M}'_n\rangle}_{(\lambda_\iota,\lambda_\kappa)=(+1,-1)}
+ \underbrace{\langle A_1^0\ldots A_n^0\rangle}_{(\lambda_\iota,\lambda_\kappa)=(+1,+1)} \geq 0
$$
Then, our result shows that
$$
\tfrac{1}{2}\langle(\mathcal{M}_n-\mathcal{M}'_n)\cdot A_{n+1}^{-1}\rangle + \tfrac{1}{2}\langle(\mathcal{M}_n+\mathcal{M}'_n) \cdot A_{n+1}^{+1}\rangle + \langle A_1^0\ldots A_{n+1}^0\rangle \geq 0
$$
is also a valid Bell inequality; rewriting the first two summands as $\mathcal{M}_{n+1}$ from the recursive definition~(\ref{MK2}), one finds that it is the $(n+1)$-partite Mermin--Klyshko inequality.
\end{expl}

\subsection{Full-correlation Bell inequalities.}

As we are now going to show, our duality Theorem~\ref{mainthm} reproduces the classification of full-correlation Bell inequalities in $(n,2,2)$-type scenarios from~\cite{WW,ZB}. This is not surprising, since our transformation~(\ref{traf}) was also used by the authors of~\cite{WW,ZB} in their proof.

Under the duality, a full-correlation Bell inequality corresponds to a full-correlator no-signaling box, i.e.~to a box of the form
$$
\langle A_1^{s_1}\ldots A_n^{s_n}\rangle = \left\{\begin{array}{cl} \varepsilon^{s_1\ldots s_n} & \textrm{if }|s_1|=\ldots =|s_n|=1 \\ 1 & \textrm{if }s_1=\ldots=s_n=0 \\ 0 & \textrm{otherwise} \end{array}\right.
$$
where the $\varepsilon^{s_1\ldots s_n}\in [-1,+1]$ are arbitrary and do not need to satisfy any additional constraints besides lying in $[-1,+1]$. Every such full-correlator box corresponds to the full-correlator Bell inequality
$$
\langle A_1^0\ldots A_n^0\rangle + 2^{-n}\sum_{\vec{s}\in\{-1,+1\}^n} \varepsilon^{s_1\ldots s_n}\sum_{\vec{t}\in\{-1,+1\}^n} F_{s_1t_1}\ldots F_{s_nt_n} \langle A_1^{t_1}\ldots A_n^{t_n}\rangle \geq 0 .
$$
We have now written the sum explicitly, as the domain of $\vec{s}$ is now not $\{-1,0,+1\}^n$, as it was before.

The extremal full-correlation no-signaling boxes are those with $\varepsilon^{\vec{s}}=\pm 1$ for all $\vec{s}\in\{-1,+1\}^n$. (Here, extremal full-correlation means extremal in the cone of full-correlation no-signaling boxes, which does not necessarily imply extremal in the cone of all no-signaling boxes.) The Bell inequalities associated to these boxes give all the facet full-correlation Bell inequalities; they can be summarized in the single nonlinear inequality
$$
\sum_{\vec{s}\in\{-1,+1\}^n}\left|\sum_{\vec{t}\in\{-1,+1\}^n} F_{s_1t_1}\ldots F_{s_nt_n} \langle A_1^{t_1}\ldots A_n^{t_n}\rangle \right| \geq 2^n \langle A_1^0\ldots A_n^0\rangle .
$$
This is precisely the main result of Werner and Wolf~\cite{WW} and \.Zukowski and Brukner~\cite{ZB} written in our notation.

\section{What about the quantum set?}
\label{quantum}

The whole incentive behind the study of no-signaling correlations arises from the counterintuitive nature of quantum correlations. The quantum correlations form a cone $Q_n^{2,2}$, lying between the Bell cone and the no-signaling cone, $B_n^{2,2}\subseteq Q_n^{2,2}\subseteq NS_n^{2,2}$, where the points of $Q_n^{2,2}$ are defined to be those which can be written in the form
$$
\langle A_1^{s_1}\ldots A_n^{s_n}\rangle = \tr\left(\rho(\mathcal{A}_1^{s_1}\otimes\ldots\otimes\mathcal{A}_n^{s_n})\right),
$$
where the $\mathcal{A}_j^s\in M_2(\C)$ are hermitian matrices satisfying $\mathcal{A}_j^s\cdot\mathcal{A}_j^s=\mathbbm{1}$ and $\mathcal{A}_j^0=\mathbbm{1}$, and $\rho\in M_{2^n}(\C)$ is any positive semidefinite matrix; in this definition, we have implicitly made use of the fact that in $(n,2,2)$ scenarios, it is sufficient to consider a qubit Hilbert space for each party~\cite{Mas}. In contrast to the standard conventions, we do not require $\rho$ to be normalized, so that the convex cone $Q_n^{2,2}$ contains precisely all the nonnegative scalar multiples of the usual convex set of quantum correlations.

We have learned in Theorem~\ref{mainthm} that there is a simple isomorphism
$$
\left(B_n^{2,2}\right)^*\stackrel{\cong}{\longrightarrow} NS_n^{2,2}
$$
identifying the Bell cone and the no-signaling cone as duals of each other. Now the obvious question is, how does $Q_n^{2,2}$ it into this picture? Is there another natural set of correlations which we can identify as the image of $\left(Q_n^{2,2}\right)^*$ under this invertible linear map? We argue now that it there is some evidence for the hypothesis that this set is $Q_n^{2,2}$ itself.

Recall that the formalism of Appendix~\ref{conesapp} means that $\left(Q_n^{2,2}\right)^*$ is the convex cone of all \emph{Tsirelson inequalities}, i.e.~the cone of all $g_{\vec{s}}$ for which
$$
g_{\vec{s}}x^{\vec{s}}\geq 0
$$
for all $x^{\vec{s}}\in Q_n^{2,2}$.

\begin{qstn}
\label{hypo}
Is the image of
\beq
\label{FQ}
\left(Q_n^{2,2}\right)^* \longrightarrow \otimes^n\R^{3} ,\qquad g_{s_1\ldots s_n}\mapsto F^{s_1t_1}\cdots F^{s_nt_n}g_{t_1\ldots t_n}
\eeq
equal to $Q_n^{2,2}$?
\end{qstn}

Given that the set of (unnormalized) quantum states, i.e.~the set of positive operators on a finite-dimensional Hilbert space, is self-dual with respect to the Hilbert--Schmidt inner product, it may be reasonable to conjecture the answer to this question to be positive. And indeed, for $n=2$, which is the first nontrivial case, this does turn out to be correct on the level of full (two-party) correlations, i.e.~upon not taking into account the marginal probabilities~\cite{PV}. As we will now see, this self-duality can be regarded as an ``explanation'' of the mysterious Tsirelson bound limiting the nonlocality of quantum correlations.

\subsection{Emergence of Tsirelson's bound.}

One peculiar property of quantum correlations is \emph{Tsirelson's bound}~\cite{Tsi}: The minimal $c\in\R$ for which the CHSH inequality~\cite{CHSH}
$$
\langle A_1^{-1}A_2^{-1} \rangle + \langle A_1^{-1}A_2^{+1} \rangle + \langle A_1^{+1}A_2^{-1} \rangle - \langle A_1^{+1}A_2^{+1} \rangle \leq c\langle A_1^0 A_2^0\rangle
$$
is valid on $Q_n^{2,2}$ is $c=2\sqrt{2}$.

We say that a set $X\subseteq\otimes^n\R^{3}$ \emph{respects the symmetries} if $X$ is invariant under the usual symmetries considered in Bell scenarios: permutations of parties; permutations of the observables of one party; permutations of the outcome of one observable of one party. If $X$ is invariant under all three kinds of symmetries, then it is automatically invariant also under any combination thereof.

\begin{prop}
Let $X\subseteq NS_2^{2,2}$ be a cone which respects the symmetries. If $X$ is self-dual in the sense that the map
\beq
\label{FX}
g_{s_1 s_2}\mapsto F^{s_1t_1} F^{s_2t_2}g_{t_1 t_2}
\eeq
is an isomorphism $X^*\stackrel{\cong}{\longrightarrow}X$, then the maximal value of $X$ in the CHSH inequality is $2\sqrt{2}$.
\end{prop}

\begin{proof}
Let $c$ be the maximal CHSH value of $X$. Then, since $X$ respects the symmetries, one can apply the depolarization technique of~\cite{MAG} in order to show that that the no-signaling box defined by
\beq
\label{Xc}
\langle A_1^{s_1} A_2^{s_2} \rangle = \begin{array}{c|ccc} s_1,\: s_2 & -1 & 0 & +1 \\\hline -1 & \tfrac{1}{4}c & 0 & \tfrac{1}{4}c \\ 0 & 0 & 1 & 0 \\ -1 & \tfrac{1}{4}c & 0 & -\tfrac{1}{4}c \end{array}
\eeq
is in $X$. A short calculation shows that applying the inverse map of~(\ref{FX}), which we know to be $x^{s_1s_2}\mapsto F_{s_1t_1}F_{s_2t_2}x^{t_1t_2}$, transforms this no-signaling box into the inequality 
$$
\tfrac{1}{8}c\langle A_1^{-1}A_2^{-1}\rangle + \tfrac{1}{8}c\langle A_1^{-1}A_2^{+1} \rangle + \tfrac{1}{8}c\langle A_1^{+1}A_2^{-1}\rangle - \tfrac{1}{8}c\langle A_1^{+1}A_2^{+1}\rangle + \langle A_1^0 A_2^0\rangle \geq 0
$$
which is, by assumption, valid on $X$, such that the coefficient $c$ in this inequality is maximal with this property. Again since $X$ respects the symmetries, one can obtain another inequality valid on $X$ by flipping the sign of the first four terms and rewriting,
$$
\tfrac{1}{8}c\langle A_1^{-1}A_2^{-1}\rangle + \tfrac{1}{8}c\langle A_1^{-1}A_2^{+1} \rangle + \tfrac{1}{8}c\langle A_1^{+1}A_2^{-1}\rangle - \tfrac{1}{8}c\langle A_1^{+1}A_2^{+1}\rangle \leq \langle A_1^0 A_2^0\rangle .
$$
Evaluating~(\ref{Xc}) on this inequality produces
$$
\tfrac{1}{32}c^2\cdot 4 \leq 1
$$
or $c\leq 2\sqrt{2}$. This shows that Tsirelson's bound holds for $X$; it remains to prove that the bound is tight.

The bound $c\leq 2\sqrt{2}$ means that the inequality
$$
\langle A_1^{-1}A_2^{-1} \rangle + \langle A_1^{-1}A_2^{+1} \rangle + \langle A_1^{+1}A_2^{-1} \rangle - \langle A_1^{+1}A_2^{+1}\rangle + 2\sqrt{2}\langle A_1^0 A_2^0\rangle \geq 0
$$
lies in $X^*$. Then the isomorphism assumption on~(\ref{FX}) implies that the box
$$
\langle A_1^{s_1} A_2^{s_2} \rangle = \begin{array}{c|ccc} s_1,\: s_2 & -1 & 0 & +1 \\\hline -1 & \tfrac{1}{\sqrt{2}} & 0 & \tfrac{1}{\sqrt{2}} \\ 0 & 0 & 1 & 0 \\ -1 & \tfrac{1}{\sqrt{2}} & 0 & -\tfrac{1}{\sqrt{2}} \end{array}
$$
lies in $X$. Therefore, $c=2\sqrt{2}$.
\end{proof}

\subsection{Extremal no-signaling boxes with quantum realization?}

It seems to be an open question whether there exists an extremal nonlocal no-signaling box which has a quantum realization. In the $(2,2,2)$ and $(3,2,2)$ scenarios, this is known not to be the case~\cite{ycats}. This lets us answer Question~\ref{hypo} in the negative.

\begin{prop}
The image of $(Q_3^{2,2})^*$ under~(\ref{FQ}) is not equal to $Q_3^{2,2}$.
\end{prop}

\begin{proof}
If the answer to Question~\ref{hypo} were positive for $n=3$, then the extremal no-signaling box \#4 of~\cite{PBS}, given by the nonvanishing expectation values
\begin{align}
\begin{split}
\label{gynibox}
\langle A_1^{-1} A_2^{+1} A_3^{0\phantom{-}}\rangle &= +1 \\
\langle A_1^{0\phantom{-}} A_2^{-1} A_3^{+1}\rangle &= +1 \\
\langle A_1^{+1} A_2^{0\phantom{-}} A_3^{-1}\rangle &= +1 \\
\langle A_1^{-1} A_2^{-1} A_3^{-1}\rangle &= +1 \\
\langle A_1^{+1} A_2^{+1} A_3^{+1}\rangle &= -1 ,
\end{split}
\end{align}
would have a quantum realization: according to Table~\ref{correspond}, this extremal box arises by applying~(\ref{FQ}) to the facet Bell inequality \#10 of~\cite{Sliwa}, which is the Guess-Your-Neighbor's-Input inequality from~\cite{GYNI}. Since this inequality has no quantum violation, it lies in $(Q_3^{2,2})^*$. If Question~\ref{hypo} had a positive answer, then the image of this inequality under~(\ref{FQ}), which is~(\ref{gynibox}), would be in $Q_3^{2,2}$. However, this is known not to be the case~\cite{ycats}.
\end{proof}

We find it a curious feature of the no-signaling box~(\ref{gynibox}) that it arises naturally in our duality studies, but also in other contexts~\cite{ycats}.

There are many other examples of facet Bell inequalities without quantum violation in $(n,2,2)$ scenarios for $n\geq 4$~\cite{UPB2}. Currently, we do not know whether the no-signaling boxes associated to these have a quantum realization or not; we expect them to be natural candidates for examples of extremal nonlocal no-signaling boxes with a quantum realization.

\bibliographystyle{plain}
\bibliography{polytope_duality}

\appendix

\section{Tensor notation}
\label{tennot}

If $V$ is a $d$-dimensional real vector space, we write $V^*$ for the dual vector space
$$
V^* = \left\{ g : V\ra\R \textrm{ linear}\right\} .
$$
It is a basic fact that $V^{**}=V$ and $(V\otimes W)^*=V^*\otimes W^*$.

We follow the tensor notation commonly used in differential geometry and general relativity~\cite{Einstein}. This means in particular that we write $x^s$ both for a generic element of $V$ and for its components in a fixed basis, with index $s=1,\ldots,d$. The tensor product of $x^s\in V$ and $y^t\in W$ then becomes $x^sy^t\in(V\otimes W)$, which makes the symbol ``$\otimes$'' obsolete in this situation. Similarly, a generic element of $V^*$, i.e.~a linear functional on $V$, is denoted by $g_s$, also with components $g_1,\ldots,g_d$. The application of the functional $g_s$ on the vector $x^s$ gives the real number $g_s x^s$; by the Einstein summation convention, the sum over $s$ is automatic here.

Concerning higher order tensors, we will only need those which are either completely covariant or completely contravariant. The notation $x^{s_1\ldots s_n}$, or $x^{\vec{s}}$, represents a generic element of $\otimes^n V$, while $g_{s_1\ldots s_n}$, or $g_{\vec{s}}$, stands for a generic element of $\otimes^n V^* = (\otimes^n V)^*$. The evaluation of $g_{\vec{s}}$ on $x^{\vec{s}}$ is $g_{\vec{s}}x^{\vec{s}}=g_{s_1\ldots s_n}x^{s_1\ldots s_n}$, where the sum is again implicit.

\section{Convex cones, tensor products, and duality}
\label{conesapp}

We use the formalism of convex cones, their duality and tensor products, a framework which has been applied, for example, to the study of general probabilistic theories~\cite{Barnum2,Barnum}. Since these techniques do not seem to be standard methods in research on quantum nonlocality, we include the following background material.

\begin{defn}
Let $V$ be a finite-dimensional real vector space. A subset $C\subseteq V$ is a \emph{cone} if it satisfies
\begin{enumerate}
\item $C+C\subseteq C$.
\item $\lambda C \subseteq C$ for $\lambda\in\R_{\geq 0}$,
\item $0\in C$,
\item $C$ is closed,
\end{enumerate} 
\end{defn}

It is a somewhat unusual but convenient definition to require all cones to be closed. We will not require $\lin(C)=V$, although all of our examples in the main text have this property. The assumption $0\in C$ is necessary in order to exclude the boring case $C=\emptyset$.

\begin{defn}
If $C\subseteq V$ is a cone, then the \emph{dual cone} $C^*$ is
$$
C^* = \left\{ g_s\in V^* \:|\: g_s x^s \geq 0 \:\:\forall x^s\in C\right\}
$$
\end{defn}

It is easy to verify that this is indeed a cone.

There are two basic ways to specify a cone. For a finite set of vectors $X\subseteq V$, the cone spanned by $X$ is $\cone(X)$, the set of all nonnegative linear combinations of elements of $X$; or, equivalently, the smallest cone containing $X$. We only consider the case of finite $X$, which is sufficient for our purposes, so that $\cone(X)$ is polyhedral.

On the other hand, for a finite set $G\subseteq V^*$, one can regard the elements of $G$ as functionals defining hyperplanes in $V$ which bound a cone. The cone defined in this way is
\beq
\label{coneS2}
\cone(G)^*=\left\{x^s\in V\:|\:\ g_s x^s \geq 0 \:\:\forall g_s\in G\right\} .
\eeq
This notation is consistent in the sense that $\cone(G)^*\subseteq V$ is indeed the dual of $\cone(G)\subseteq V^*$, since $g_s x^s \geq 0\:\forall g_s\in G$ is equivalent to $g_s x^s \geq 0\:\forall g_s\in\cone(G)$.

\begin{lem}
Let $C\subseteq V$ and $D\subseteq V^*$ be cones.
\label{basic}
\begin{enumerate}
\item\label{galoisconnection} $D\subseteq C^*$ if and only if $C\subseteq D^*$.
\item\label{basicdual} $x^s\in C \Leftrightarrow g_s x^s \geq 0\:\:\forall g_s\in C^*$.
\item\label{doubledual} $C^{**}=C$.
\end{enumerate}
\end{lem}

\begin{proof}
\begin{enumerate}
\item $D\subseteq C^*$ means that $g_s x^s \geq 0$ for all $x^s\in C$ and $g_s\in D$; so does $C\subseteq D^*$.
\item If $x^s\in C$, then $g_s x^s\geq 0\:\:\forall g_s\in C^*$ by definition of $C^*$. On the other hand, if $x^s\notin C$, then the Hahn-Banach separation theorem guarantees the existence of a $g_s\in C^*$ with $g_s x^s <0$.
\item Together with the identification $V^{**}=V$, this is a restatement of part~\ref{basicdual}.
\end{enumerate}
\end{proof}

\begin{defn}
A vector $x^s\in C$ is an \emph{extreme ray} if, for any decomposition $x^s=x'^s+x''^s$ with $x'^s,x''^s\in C$, there is $\lambda\in\R_{\geq 0}$ with $x^s=\lambda x'^s$.
\end{defn}

We write $\mathrm{ex}(C)$ for the set of extreme rays of $C$. 

\begin{prop}
\label{km}
$C=\cone(\ex(C))$.
\end{prop}

\begin{proof}
This is the finite-dimensional Krein--Milman theorem.
\end{proof}

\subsection{Tensor products of cones.}

Tensor products of cones were introduced by Namioka and Phelps~\cite{NP}. The maximal and minimal tensor product, which we now define, are commonly used in the formalism of general probabilistic theories~\cite{Barnum2,Barnum}.

For cones $C\subseteq V$ and $D\subseteq W$, the \emph{maximal tensor product} is the cone
\beq
\label{defomax}
C\omax D = \left\{ z^{st}\in\left( V\otimes W\right) \:|\: g_{s} h_{t}z^{st}\geq 0 \:\:\: \forall g_s\in C^*,\: h_t\in D^* \right\} .
\eeq
Thanks to Proposition~\ref{km}, it is enough to check the condition $g_{s}h_{t}z^{st}\geq 0$ only for $g_s\in\ex(C^*)$ and $h_t\in\ex(D^*)$.

It is easy to see that $\omax$ is an associate binary operation on cones. This observation makes the $\omax$-product of several cones $C_1\subseteq V_1,\ldots C_n\subseteq V_n$ well-defined. It is given by
$$
C_1\omax\ldots\omax C_n = \left\{ z^{\vec{s}}\in \left( V_1\otimes\ldots\otimes V_n\right) \:|\: g_{1,s_1} \cdots g_{n,s_n}z^{s_1\ldots s_n}\geq 0 \:\:\:\forall g_{i,s}\in C^*_i\right\} .
$$

Similarly, the \emph{minimal tensor product} of $C$ and $D$ is
$$
C\omin D = \cone\left\{ c^{s}d^{t} \::\: c^s\in C,\: d^t\in D \right\}\subseteq  V\otimes  W
$$
which is also an associative operation on cones, with its $n$-fold version given by
$$
C_1\omin\ldots\omin C_n = \cone\left\{c_1^{s_1}\cdots c_n^{s_n}\::\: c_i^s\in C_i \right\}  \subseteq V_1\otimes\ldots\otimes V_n .
$$

The tensor product $\omin$ is minimal in the sense that it is, by definition, the smallest cone containing all tensor products of elements of its factors.

\begin{lem}
The maximal and minimal tensor products are dual in the following sense:
\label{tensordual}
\begin{align}
\label{minmaxdual} (C_1\omin\ldots\omin C_n)^* &= C^*_1\omax\ldots\omax C^*_n \\
\label{maxmindual} (C_1\omax\ldots\omax C_n)^* &= C^*_1\omin\ldots\omin C^*_n 
\end{align}
\end{lem}

\begin{proof}
A $g_{\vec{s}}\in\left(V_1^*\otimes\ldots\otimes V_n^*\right)$ lies in $(C_1\omin\ldots\omin C_n)^*$ whenever $g_{s_1\ldots s_n}c_1^{s_1}\cdots c_n^{s_n}\geq 0$ for all $c_i^s\in C_i$. On the other hand, it lies in $C^*_1\omax\ldots\omax C^*_k$ whenever $c_1^{s_1}\cdots c_n^{s_n} g_{s_1\ldots s_n} \geq 0$ for all $c_i^s\in C_i^{**}=C_i$, which is the same condition; this shows~(\ref{minmaxdual}).

Then~(\ref{maxmindual}) follows from~(\ref{minmaxdual}) upon replacing each $C_i$ by $C_i^*$ and applying Lemma~\ref{basic}.
\end{proof}

\begin{expl}
Let $M_k(\C)_h$ be the vector space of all hermitian complex $k\times k$-matrices. The positive semidefinite cone
$$
\mathcal{S} = \left\{ \rho\in M_k(\C)_h \:|\: \rho \geq 0 \right\}
$$
is the set of all (unnormalized and mixed) quantum states on $\C^k$. Then, by the definition of separability~\cite{Werner}, $\mathcal{S}\omin\mathcal{S}$ is precisely the set of (unnormalized) separable states on $\C^k\otimes\C^k$; on the other hand, $\mathcal{S}\omax\mathcal{S}$ is the corresponding set of entanglement witnesses.
\end{expl}

\begin{lem}
\label{maxcrit}
For $z^{st}\in V\otimes W$,
$$
z^{st}\in C\omax D \Longleftrightarrow g_s z^{st}\in D \:\:\forall g_s\in \ex(C^*).
$$
\end{lem}

\begin{proof}
By Lemma~\ref{basic}.\ref{basicdual}, the condition $g_s z^{st}\in D$ is equivalent to $g_s h_t z^{st}\geq 0$ for all $h_t\in D^*$. With this, the assertion follows from the definition~(\ref{defomax}).
\end{proof}

\end{document}